\documentclass{article}
\usepackage[numbers]{natbib}

\usepackage{graphicx} 
\usepackage{amsmath}
\usepackage{amssymb}
\usepackage{amsthm}
\usepackage{enumitem}
\usepackage{authblk}

\newtheorem{theorem}{Theorem}
\newtheorem{lemma}{Lemma}

\title{Charged black holes surrounded by shells of charged particles}
\author{D.~Fajman\footnote{Faculty of Physics, University of Vienna, Austria}, H.~Seer}

\date{\today}

\begin{document}

\maketitle

\begin{abstract}
We construct new static solutions to the Einstein-Vlasov-Maxwell system. These solutions model Reissner-Nordstr\"om-type black holes with surrounding shells of charged matter.
\end{abstract}

\section{Introduction}
The Einstein-Vlasov-Maxwell system models an ensemble of collisionless, charged, self-gravitating particles in general relativity. The construction of static solutions of that system serves as a first step to understand and classify potentially stable configurations 
composed of such particles. Many classes of solutions have been constructed in recent years with various features modelling objects that range from galaxy clusters to photon shells \cite{RR93, R94,AFT17}, even stationary, non-static models \cite{AKR11,AKR14}. In the particular case of charged particles and the presence of electromagnetic fields the corresponding solutions to the Einstein-Vlasov-Maxwell system  
model for instance galactic nebulae, which consist of ionized gas. 

Static solutions with charged particles were constructed first by Thaller and later also generalized to the stationary case \cite{thaller,T20}. These solutions have regular centres. It is however expected that there exist nebula clusters, which surround black holes \cite{SUB-12}. In this paper we construct rigorous examples of static solutions to the Einstein-Vlasov-Maxwell system with Reissner-Nordstr\"om black holes at the center as first potential models for this configuration.\\
Other static solutions with black holes were constructed for the case of uncharged
particles in \cite{R94,AFT15,A21, J21}.
The standard approach to construct static solutions is based on a spherically symmetric ansatz for the metric and a compatible ansatz for the distribution function. In the particular case here we use a perturbative approach based on static solutions for the uncharged system \cite{R94} with a black hole at the center. A major difference to the uncharged case is that the key equation to determine the location of the inner radius of the shell takes the form of a third order polynomial equation in the presence of charges, which requires a different analysis to the uncharged case. 
\subsubsection*{Acknowledgements}
This research was funded in whole or in part by the Austrian Science Fund (FWF) [\emph{Matter dominated Cosmology} 10.55776/PAT7614324].

%%%%%%%%%%%
%%%%%%%%%%%

\section{Setup}
\label{chapter setup}

We consider the static, spherical symmetric Einstein-Vlasov-Maxwell system (EVMS) with vanishing cosmological constant $\Lambda = 0$. This follows the standard setup as for instance in \cite{thaller}. We use similar notations and introduce the approach for the sake of completeness. It is based on the Ansatz

\begin{equation}
    ds^2 = - e^{2\mu (r)} dt^2 + e^{2\lambda (r)} dr^2 + r^2 \left( d \Theta^2 + \sin^2{\Theta} d \varphi^2 \right)
\end{equation}

for the metric, where $\mu(r)$, $\lambda(r) \in \mathbb{R}$. Furthermore, we consider all particles to be of the same mass and charge (given by $q_0$) and their distribution function $f$ to be a function only of the radius $r$, the radial momentum $\omega$ and the angular momentum $L$. We can write the relevant matter quantities, i.e. the energy density $\rho(r)$, the radial pressure $p(r)$, the transversal pressure $p_T(r)$ and the charge density $\rho_q(r)$ in terms of integrals over the particle distribution function:

\begin{align}
    \rho (r) &= \frac{\pi}{r^2} \int_0^\infty \int_{-\infty}^\infty f (r, \omega, L) \sqrt{1 + \omega^2 + \frac{L}{r^2}} d \omega dL \label{rho}\\
    p(r) &= \frac{\pi}{r^2} \int_0^\infty \int_{-\infty}^\infty f(r, \omega, L ) \frac{\omega^2}{\sqrt{1+\omega^2+\frac{L}{r^2}}} d \omega dL \label{p pressure}\\
    p_T(r) &= \frac{\pi}{2r^4} \int_0^\infty \int_{-\infty}^\infty f(r, \omega, L ) \frac{L}{\sqrt{1+\omega^2+\frac{L}{r^2}}} d \omega dL \label{p trans momentum}\\
    \rho_q(r) &= q_0 e^{\lambda (r)} \frac{\pi}{r^2} \int_0^\infty \int_{-\infty}^\infty f(r, \omega, L) d \omega dL \label{rho_q}
\end{align}

With our symmetry considerations the Vlasov equation looks like this:

\begin{equation} \label{Vlasov}
    \omega \frac{\partial f}{\partial r} + \left( \frac{L}{r^3} + q_0 \frac{q(r)}{r^2} e^{\mu(r)+\lambda(r)} - \mu'(r) \left(1+w^2+ \frac{L}{r^2} \right) \right) \frac{\partial f}{\partial \omega} = 0
\end{equation}

where $q(r)$ is the charge contained in a sphere of radius $r$ around the center, determined by the only non-trivial Maxwell equation of the static, spherical symmetric system

\begin{equation}
    q'(r) = 4 \pi r^2 \rho_q(r) \label{Maxwell}
\end{equation}

There are two independent, conserved quantities along the characteristics of the Vlasov equation, namely $L$ and the energy

\begin{equation}
    E = e^{\mu(r)} \sqrt{1 + \omega^2 + \frac{L}{r^2}} - \Tilde{I}_q(r)
\end{equation}

where

\begin{equation} \label{I_Lambda tilde}
    \Tilde{I}_q(r) := q_0 \int_{r_v}^r \frac{q_\Lambda(s)}{s^2} e^{\mu_\Lambda(s)+\lambda_\Lambda(s)} ds
\end{equation}

can be seen as an effective electric energy and $r_v$ is a radius in the vacuum region outside of the black hole horizon. Any Ansatz $f(r, \omega, L) = \Phi(E,L)$ solves the Vlasov equation. For convenience we define

\begin{equation}
    \epsilon := \sqrt{1 + \omega^2 + \frac{L}{r^2}}
\end{equation}.

and use in the following the polytropic Ansatz

\begin{equation}
    f(r, \omega, L) = \Phi (E,L) = \alpha \left[ 1- \frac{E}{E_0}\right]^k_+ \left[ L-L_0 \right]^l_+ = \alpha \left[ 1- \frac{\epsilon e^\mu - \Tilde{I}_q}{E_0}\right]^k_+ \left[ L-L_0 \right]^l_+ \label{Ansatz}
\end{equation}

where $\alpha, k \geq 0$,\: $l> - \frac{1}{2}$,\: $k< 3l + \frac{7}{2}$,\: $E_0 > 0$ and $L_0 \geq 0$ and

\begin{equation*}
    [x]_+ = \left\{ \begin{array}{cc}
         x\text{\: , for }x \geq 0&  \\
         0\text{\: , for }x < 0&
         \end{array} \right.
\end{equation*}

As argued in \cite{thaller}, we can set $E_0=1$ by rescaling the time coordinate. With this Ansatz in \eqref{Ansatz} we can write the matter quantities \eqref{rho}, \eqref{p pressure} and \eqref{rho_q} as

\begin{align}
    \rho(r) &= g(r,\mu(r), \Tilde{I}_q(r)) \label{rho g}\\
    p(r) &= h(r,\mu(r), \Tilde{I}_q(r)) \label{p h}\\
    \rho_q(r) &= q_0 e^{\lambda(r)} k(r, \mu(r), \Tilde{I}_q(r)) \label{rho_q k}
\end{align}

with functions

\begin{align}
    g(r,\mu, \Tilde{I}_q) &= \alpha c_l r^{2l} \int_{\sqrt{1 + \frac{L_0}{r^2}}}^{e^{-\mu}(E_0 + \Tilde{I}_q)} \left( 1 - \frac{\epsilon e^\mu - \Tilde{I}_q}{E_0} \right)^k \epsilon^2 \left( \epsilon^2 - \left(1 + \frac{L_0}{r^2} \right) \right)^{l+\frac{1}{2}} d\epsilon \label{g}\\
    h(r,\mu, \Tilde{I}_q) &= \frac{\alpha c_l}{2l+3} r^{2l} \int_{\sqrt{1 + \frac{L_0}{r^2}}}^{e^{-\mu}(E_0+\Tilde{I}_q)} \left( 1 - \frac{\epsilon e^\mu - \Tilde{I}_q}{E_0} \right)^k \left( \epsilon^2 - \left(1 + \frac{L_0}{r^2} \right) \right)^{l+\frac{3}{2}} d\epsilon \label{h}\\
    k(r,\mu, \Tilde{I}_q) &= \alpha c_l r^{2l} \int_{\sqrt{1 + \frac{L_0}{r^2}}}^{e^{-\mu}(E_0+\Tilde{I}_q)} \left( 1 - \frac{\epsilon e^\mu - \Tilde{I}_q}{E_0} \right)^k \epsilon \left( \epsilon^2 - \left(1 + \frac{L_0}{r^2} \right) \right)^{l+\frac{1}{2}} d\epsilon \label{k}
\end{align}

with

\begin{equation}
    c_l = 2\pi \int_0^1 \frac{s^l}{\sqrt{1-s}} ds
\end{equation}.

Here we understand $g=h=k=0$ if in the integrals in \eqref{g}-\eqref{k} the upper integral limits are smaller or equal the lower integral limits, that is, for

\begin{equation} \label{vac cond}
    \frac{\sqrt{1 + \frac{L_0}{r^2}} e^\mu - \Tilde{I}_q}{E_0} \geq 1
\end{equation}

We can express this condition also as

\begin{equation} \label{gamma}
    0 \geq \ln(E_0+ \Tilde{I}_q) - \frac{1}{2} \ln\left( 1 + \frac{L_0}{r^2} \right) -\mu =: \gamma(r)
\end{equation}

where we defined the new quantity $\gamma(r)$. If on the other hand, the upper integral limits are greater than the lower ones, the integrals in \eqref{g}-\eqref{k} and therefore also the matter quantities $\rho(r)$, $p(r)$ and $\rho_q(r)$ will be positive. This means that we have vacuum, in the sense that all the matter quantities are zero, at a radius $r$ if and only if $\gamma(r) \leq 0$. We note further some properties of the three functions $g$, $h$ and $k$ in \eqref{g}-\eqref{k}. At first, since for positive $g$ and $h$ we have in their definition in \eqref{g} and \eqref{h} that $\epsilon \geq \sqrt{1+ \frac{L_0}{r^2}} \geq 1$, we can conclude that

\begin{equation}
    h \leq \frac{g}{2l+3} \label{gh ineq}
\end{equation}

Furthermore, we state the following lemma

\begin{lemma}The functions $g(r,\mu,\Tilde{I}_q)$, $h(r,\mu,\Tilde{I}_q)$ and $k(r,\mu, \Tilde{I}_q)$ defined in \eqref{g}-\eqref{k} have the following properties:

\begin{enumerate}[label=\roman*.] \label{properties g,h,k}
    \item The functions are continuously differentiable
    \item The functions and all their partial derivatives are increasing in the first and third argument
    \item The functions and their partial derivatives with respect to the first and third argument are non-increasing in the second argument
    \item Their partial derivatives with respect to the second argument are non-decreasing in the second argument.
\end{enumerate}

\end{lemma}

Lemma \ref{properties g,h,k} is proven in \cite{thaller} and we will need it later, in the course of the proof of the existence of solutions. In the considered system we get two independent, non-trivial Einstein equations, which read

\begin{align}
    e^{-2 \lambda (r)} \left(2r \lambda'(r) - 1 \right) + 1 &= 8 \pi r^2 \rho (r) + \frac{q^2(r)}{r^2} \label{lambda diff}\\
    e^{-2 \lambda (r)} \left(2r \mu'(r) + 1 \right) - 1 &= 8 \pi r^2 p(r) - \frac{q^2(r)}{r^2} \label{mu diff}
\end{align}

\section{Main Theorem: Charged black holes surrounded by charged shells}
\label{chapter BH Lambda=0}

We state the main theorem of the paper in the following.

\begin{theorem} [Main Theorem]  \label{lemma background}
There exists an $\varepsilon>0$ such that for all $(q_0,Q_0)$ of identical sign with $|q_0|+|Q_0|<\varepsilon$ and for all $L_0>0$, $M_0>0$ such that $M_0^2 > Q_0^2$ and $L_0 > 16 M_0^2$ there exists a static spherically symmetric spacetime with radial variable $r>0$ of the following form.
For $r\in(0,r_{\mathrm{shell}})$, where $r\in(r_{\mathrm{shell}},R_{\mathrm{shell}})$ is strictly bigger than the event horizon of the black hole, $M_0+\sqrt{M_0^2-Q_0^2}$, the spacetime is a Reissner-Nordstr\"om black hole with mass $M_0$ and charge $Q_0$,

\begin{equation}
-(1-\frac{2M_0}{r}+\frac{Q_0^2}{r^2})dt^2+(1-\frac{2M_0}{r}+\frac{Q_0^2}{r^2})^{-1}dr^2+dS_r,
\end{equation}
where $dS_r$ denotes the surface-element of the sphere of radius $r$.
\indent For $r\in(r_{\mathrm{shell}}, R_{\mathrm{shell}})$ the metric takes the form
\begin{equation}
- e^{2\mu_q (r)} dt^2 + e^{2\lambda_q (r)} dr^2 + r^2 \left( d \Theta^2 + \sin^2{\Theta} d \varphi^2 \right),
\end{equation}
where $(\lambda_q,\mu_q)$ is a solution to the reduced Einstein equations \eqref{lambda diff}, \eqref{mu diff} and the matter density $\rho(r)$ and the charge density $\rho_q(r)$ are non-vanishing for $r\in(r_{\mathrm{shell}},R_{\mathrm{shell}})$.\\
For $r\in[R_{\mathrm{shell}},\infty)$ the metric is a Reissner-Nordstr\"om metric of the form 

\begin{equation}
-(1-\frac{2M_{\mathrm{total}}}{r}+\frac{Q_{\mathrm{total}}^2}{r^2})dt^2+(1-\frac{2M_{\mathrm{total}}}{r}+\frac{Q_{\mathrm{total}}^2}{r^2})^{-1}dr^2+dS_r,
\end{equation}
where $M_{\mathrm{total}}=M_0+M_{\mathrm{shell}}$ and $Q_{\mathrm{total}}=Q_0+Q_{\mathrm{shell}}$.

\end{theorem}

\subsection{Proof of the main theorem}
\subsubsection{Strategy of the proof}
 As proven in \cite{rein} there exist global solutions $(\mu(r), \lambda(r))$ to the uncharged EVS with $\Lambda =0$ with a black hole singularity at the center and the matter quantities are supported on a shell $\{r_{+} < r <  R_0 \}$. We make use of this result to show that, introducing sufficiently weak charges to the system, so staying sufficiently close to the uncharged solution, there still exist global solutions with non-trivial matter quantities of finite spatial support. In the following let for each quantity $x$ in the uncharged case denote $x_q$ the corresponding quantity for the charged solution to the EVMS. \\

We divide the proof of the theorem into a few lemmas to simplify its organization. In the first lemma we want to show that under the conditions of Theorem \ref{lemma background} there exists a region $[r_{-q}, r_{+q}]$ outside the black hole horizon where a nontrivial Ansatz \eqref{Ansatz} where $E_0=1$ (assumed from now on, without loss of generality) gives a vacuum solution, more precisely a part of the outer region of the Reissner-Nordstr\"om solution with parameters $(M_0,Q_0)$ of the black hole.

\begin{lemma}
Let $\varepsilon>0$ be sufficiently small. Then for $(q_0,Q_0)$ of identical sign with $|q_0|+|Q_0|<\varepsilon$ and for all $L_0>0$, $M_0>0$ such that $M_0^2 > Q_0^2$ and $L_0 > 16 M_0^2$ there exists a non-trivial region $[r_{-q}, r_{+q}]$ outside the black hole event horizon where a nontrivial Ansatz \eqref{Ansatz} yields a vacuum solution
  \begin{align}
        e^{2 \mu_q(r)} &= 1- \frac{2 M_0}{r} + \frac{Q_0^2}{r^2} \label{RN mu}\\
        e^{2 \lambda_q(r)} &= \left( 1- \frac{2 M_0}{r} + \frac{Q_0^2}{r^2} \right) ^{-1}.
    \end{align}

\end{lemma}

\begin{proof}

In a vacuum region outside of the black hole horizon the effective electromagnetic energy $\Tilde{I}_{qv}$ reads

    \begin{equation} \label{I_q tilde in vac}
        \Tilde{I}_{qv}(r) = q_0 \int_{r_v}^r \frac{q(s)}{s^2}ds = q_0 q(r_v) \left(\frac{1}{r_v}- \frac{1}{r} \right) = q_0 Q_0 \left(\frac{1}{r_v}- \frac{1}{r} \right)
    \end{equation}  
    such that the vacuum condition \eqref{vac cond} then reads

    \begin{equation} \label{vac cond BH}
        \sqrt{1 + \frac{L_0}{r^2}} \sqrt{1- \frac{2M_0}{r} + \frac{Q_0^2}{r^2}} - q_0 Q_0 \left( \frac{1}{r_v} - \frac{1}{r}\right) \geq 1
    \end{equation}
    
    We note that $\Tilde{I}_{qv}(r)$ is bigger than zero, because $q_0$ and $Q_0$ have the same sign and $r$ is bigger or equal $r_v$ and it is trivially bounded from above by

    \begin{equation} \label{Tilde(I)_qv estimate}
        \Tilde{I}_{qv}(r) \leq \frac{q_0 Q_0}{r_v},
    \end{equation}
    
    which will be used later. Now we define the function
    
    \begin{equation}
        a_{Q_0}(r):= \sqrt{1- \frac{2 M_0}{r} + \frac{Q_0^2}{r^2}} \cdot \sqrt{1+\frac{L_0}{r^2}},
    \end{equation}
which appears in the vacuum condition above and the corresponding function for vanishing black hole charge by
    
    \begin{equation}
        a(r):= \sqrt{1- \frac{2 M_0}{r}} \cdot \sqrt{1+\frac{L_0}{r^2}}.
    \end{equation}

%%%%

    Basically, we want to find the following configuration which is exemplarily depicted in figure \ref{radii config Lambda=0}\par
    \begin{figure}[h!]
        \centering
        \includegraphics[width=\textwidth]{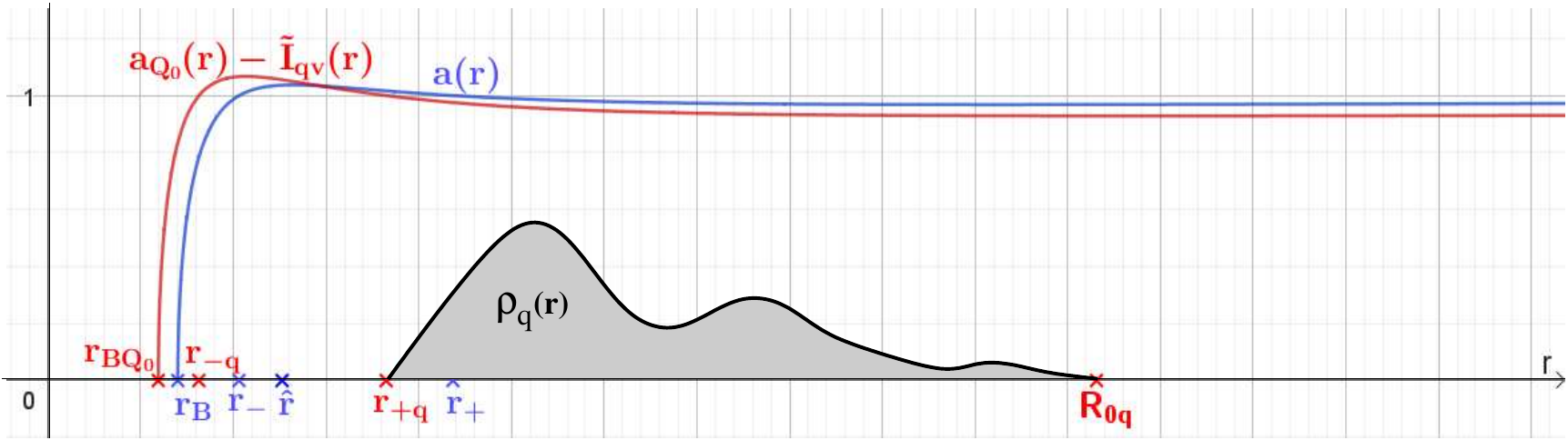} % Bild einfügen
        \caption{Sketch of an exemplary configuration of a weakly charged black hole surrounded by a weakly charged matter shell for $\Lambda = 0$}
        \label{radii config Lambda=0} % Referenz, um auf das Bild zu verweisen
    \end{figure}
      On the x-axis the radius is plotted, while the y-axis is dimensionless. We see in blue the function $a(r)$ for the background solution and in red the function $a_{Q_0}(r) - \Tilde{I}_{qv}(r)$. In the region where the respective function is bigger than one, there is a vacuum region for the corresponding solution. $r_{BQ_0}$ and $r_B$ are the black hole horizons for the respective solution, $r_{-}$ and $r_{-q}$ the first radius outside of the black hole, where the function takes the value one, $r_{+}$ and $r_{+q}$ are the second ones outside of the black hole, $\hat{r}$ is the radius, where $a(r)$ attains its maximum and $R_{0q}$ is the boundary of the support of the matter quantities of the charged solution. In the figure $\rho_q(r)$ is depicted to exemplify the matter quantities, as well corresponding to the charged solution.

     To show this configuration we define further the radicands of the functions above by

    \begin{equation}
        v(r):= 1- \frac{2 M_0}{r}
    \end{equation}
   and  
    \begin{equation} \label{v_Q0}
        v_{Q_0}(r):= 1- \frac{2 M_0}{r} + \frac{Q_0^2}{r^2}
    \end{equation}

    and notice that the zeros of $v(r)$ and $v_{Q_0}(r)$ coincide with the zeros of $a(r)$ and $a_{Q_0}(r)$ respectively and since $v(r)$ and $v_{Q_0}(r)$ correspond to the first metric components of the Schwarzschild and the Reissner-Nordstr\"om metric respectively, their zeros give rise to the horizons of the corresponding system. \\
    Note that if $v_{Q_0}(r)<0$ then $a_{Q_0}(r)$ is not defined (as a real-valued function).

    For the background solution it is known that the black hole horizon $r_B$ is given by $r_B=2M_0$ and under the condition $L_0 > 16M_0^2$ there exist two radii $r_{\pm}$ where $a(r)=1$ between which $a(r)>1$ and we have $r_B < r_{-} < r_{+}$. (see \cite{fajman}, p. 2676)\\ 
    
     The condition $M_0^2>Q_0^2$ assures that $a_{Q_0}$ has two positive, real zeros, which follows from the explicit solution for second order polynomials. The black hole horizon for the solution to the charged EVMS corresponds to the bigger zero and is given by

    \begin{equation} \label{r_BQ0}
        r_{BQ_0} = M_0 + \sqrt{M_0^2-Q_0^2}.
    \end{equation}

    From there we see immediately that $r_{BQ_0} < r_B$.
    
    We analyze now the location of the values of $r$ where $a_{Q_0}(r)=1$. The equation $a_{Q_0}(r)=1$ for finite $r$ is equivalent to the cubic equation
    
    \begin{equation}
        -2M_0r^3 + \left(Q_0^2 + L_0 \right) r^2 -2M_0L_0r + Q_0^2L_0 = 0
    \end{equation}
    
The sign of its discriminant
    
    \begin{equation}
    \small
        D = \frac{72 (Q_0^2 + L_0) Q_0^2 M_0^2 L_0^2 - 64 M_0^4 L_0^3 - 108 Q_0^4 M_0^2 L_0^2 + 4 (Q_0^2 + L_0)^2 M_0^2 L_0^2 -4 (Q_0^2+L_0)^3 Q_0^2 L_0}{1728 M_0^4}
    \end{equation}
  
determines how many different real solutions the cubic equation has (see \cite{bosch}). As we show in the following the sign is positive. 

 Since we only need to find out the sign of the discriminant, we can break it down to the following expression
    
    \begin{equation} \label{discriminant}
        \Tilde{D} := 4M_0^2L_0 \left(-2Q_0^4 + 5Q_0^2L_0-4M_0^2L_0 \right) +M_0^2L_0^3 - Q_0^2 \left(Q_0^2+L_0 \right)^3
    \end{equation}
    
    by dropping a factor $\frac{4L_0}{1728 M_0^4}$ which has a positive sign. Imposing the conditions $M_0^2>Q_0^2$ and $L_0>16M_0^2$ one can show that the expression in \eqref{discriminant} is strictly positive (cf.~subsection \ref{A1} of the appendix). In consequence, there exist three distinct real solutions to the cubic equation. \\

There are three finite, real values $r_i$ with $a_{Q_0}(r_i)=1$ with $i=1,2,3$, given by the following expressions (see \cite{bosch}) 

    \begin{align}
    \begin{split} \label{cubic sol}
        r_1 &= \frac{Q_0^2+L_0}{6 M_0} + \sqrt{- \frac{4p}{3}} \cos{\left( \frac{1}{3} \arccos{\left(-\frac{q}{2} \sqrt{-\frac{27}{p^3}}\right)} \right)}\\
        r_{2,3} &= \frac{Q_0^2+L_0}{6 M_0} - \sqrt{- \frac{4p}{3}} \cos{\left( \frac{1}{3} \arccos{\left(-\frac{q}{2} \sqrt{-\frac{27}{p^3}}\right)} \pm \frac{\pi}{3} \right)}
    \end{split}
    \end{align}
    
    with 
    
    \begin{equation} \label{p}
        p = L_0 - \frac{\left(Q_0^2+L_0\right)^2}{12M_0^2}
    \end{equation}
    
    and
    
    \begin{equation}
        q = \frac{L_0(Q_0^2+L_0)}{6M_0} - \frac{Q_0^2L_0}{2M_0}- \frac{Q_0^6+3Q_0^4L_0+3Q_0^2L_0^2+L_0^3}{108M_0^2}.
    \end{equation}

    Since

    \begin{equation} \label{a(0)}
        \underset{r \searrow 0}{\lim} a_{Q_0}(r) = \infty
    \end{equation}

    and $a_{Q_0}(r)$ is a continuous function as long as $v_{Q_0}(r)>0$, the smallest of the three solutions \eqref{cubic sol} is smaller than the first zero of $a_{Q_0}(r)$ and therefore also smaller than the event horizon $r_{BQ_0}$, so it will not be relevant for our solution. \\
    
    Hence, we focus on the two larger radii of the set $\{r_1,r_2,r_3\}$. We call the second largest one $r_{-Q_0}$ and the largest one $r_{+Q_0}$.  Elementary analysis shows that of the following three expressions
     \begin{gather}
        \cos{x} \label{cos1}\\
        - \cos{\left(x - \frac{\pi}{3} \right)}\\
        - \cos{\left(x + \frac{\pi}{3} \right)} \label{cos3}
    \end{gather}
    at any $x$ the second lowest expression is always bigger equal $- \frac{1}{2}$. This implies the estimate
    
    \begin{equation}
        r_{-Q_0}\geq \frac{Q_0^2+L_0}{6M_0} - \frac{1}{2} \sqrt{-\frac{4p}{3}}.
    \end{equation}
    
    Plugging in the expression \eqref{p} for $p$ we get
    
    \begin{equation} \label{for apx A2}
        r_{-Q_0}\geq \frac{Q_0^2+L_0-\sqrt{\left(Q_0^2+L_0\right)^2-12L_0M_0^2}}{6M_0}
    \end{equation}
    
    it can be shown (see appendix subsection \ref{A2}) that
    
    \begin{equation} \label{rhs > r0}
        \frac{Q_0^2+L_0-\sqrt{\left(Q_0^2+L_0\right)^2-12L_0M_0^2}}{6M_0} > M_0-\sqrt{M_0^2-Q_0^2}
    \end{equation}
    
    so in combination we have
    
    \begin{equation}
        r_{-Q_0} > M_0 - \sqrt{M_0^2 - Q_0^2}
    \end{equation}

    meaning that $r_{-Q_0}$ is bigger than the first zero of $a_{Q_0}(r)$. Furthermore, since $v_{Q_0}(r)$ is continuous, between its two zeros it is either strictly positive or strictly negative, so we only have to check one point between the two zeros which we choose to be the average of the two zeros $\frac{M_0 + \sqrt{M_0^2-Q_0^2} + M_0 - \sqrt{M_0^2-Q_0^2}}{2} = M_0$ and we obtain

    \begin{equation}
        v(M_0) = 1-\frac{2M_0}{M_0}+ \frac{Q_0^2}{M_0^2} = -1 + \frac{Q_0^2}{M_0^2} < 0 
    \end{equation}

    where we used the condition $M_0^2 > Q_0^2$. So, $v_{Q_0}(r)$ is strictly negative between its two zeros, that is for $M_0-\sqrt{M_0^2-Q_0^2} < r < r_{BQ_0}$. 
    
    Therefore, we can conclude that

    \begin{equation}
        r_{-Q_0} > r_{BQ_0}
    \end{equation}
    
    There is also a lower bound on the largest of the three cosines in \eqref{cos1} - \eqref{cos3} which is given by $+ \frac{1}{2}$. From that, we can derive a lower bound on $r_{+Q_0}$ and find

    \begin{equation}
        r_{+Q_0} \geq r_{-}
    \end{equation}
    
    (see appendix subsection \ref{A3} for the calculation). In the interval $(r_{-},r_{+})$, we have $a(r)>1$. Since $a_{Q_0}(r) \geq a(r)$ at any $r$ (where $a_{Q_0}(r)$ is defined),  $a_{Q_0}(r)>1$ holds in this region. Furthermore, due to continuity of $a_{Q_0}(r)$, we get

    \begin{equation}
        r_{-Q_0} < r_{-} < r_{+} < r_{+Q_0}
    \end{equation}
 This is the desired configuration we intended to show.\\

    In the end we want to find a region $[r_{-q},r_{+q}]$, where we can continue the initial Ansatz $f_q = 0$ by a nontrivial Ansatz \eqref{Ansatz}, i.e. a region where also the nontrivial Ansatz yields vacuum. To do so, we still need to take the second term in condition \eqref{vac cond BH}, i.e. the $\Tilde{I}_{qv}(r)$-term, into account. We remember that inside of a vacuum region $[r_{-q}, r_{+q}]$ we have the estimate \eqref{Tilde(I)_qv estimate} for the electromagnetic energy contribution. With this in mind, we estimate the difference between $a(r)$ and $a(r)-\Tilde{I}_{qv}(r)$ at the unique maximum $\hat{r}$ of $a(r)$ between $r_{-}$ and $r_{+}$ to

    \begin{equation} \label{dist a, a-I_qv}
        a(\hat{r}) - \left(a(\hat{r}) - \Tilde{I}_{qv}(\hat{r}) \right) = \Tilde{I}_{qv}(\hat{r}) <\frac{q_0 Q_0}{r_v}
    \end{equation}
    
    Thus, in this step it is sufficient to choose either the absolute value of $q_0$ or of $Q_0$ small enough, such that $a(\hat{r})-(a(\hat{r})-\Tilde{I}_{qv}(\hat{r})) \leq a(\hat{r}) - 1$, which means that $a(\hat{r}) - \Tilde{I}_{qv}(\hat{r}) > 1$ and guarantees the existence of exactly two ones $\{r_a, r_b\}$ of the function $a(r) - \Tilde{I}_{qv}(r)$ between $r_{-}$ and $r_{+}$. Since $a(r) - \Tilde{I}_{qv}(r) \leq a_{Q_0}(r) - \Tilde{I}_{qv}(r) \leq a_{Q_0}(r)$ for all $r$, $a_{Q_0}(r)-\Tilde{I}_{qv}(r)$ has exactly two ones $r_{-q}$ and $r_{+q}$ in the interval $[r_{-Q_0}, r_{+Q_0}]$ such that
    
    \begin{equation}
        r_{-Q_0} \leq r_{-q} \leq r_a \leq r_b \leq r_{+q} \leq r_{+Q_0}
    \end{equation}
    
    and $a_{Q_0}(r) - \Tilde{I}_{qv}(r) >1$ between $r_{-q}$ and $r_{+q}$. In this region $[r_{-q}, r_{+q}]$ we have

    \begin{equation}
        a_{Q_0}(r) - \Tilde{I}_{qv}(r) = e^{\mu_{q}(r)} \sqrt{1+ \frac{L_0}{r^2}} - \Tilde{I}_{qv}(r) > 1
    \end{equation}

    with $e^{\mu_q(r)} = \sqrt{1-\frac{2M_0}{r} + \frac{Q_0^2}{r^2}}$ which is precisely the condition \eqref{vac cond BH} for having a vacuum.
  \end{proof}

  \begin{lemma}
  For a nontrivial Ansatz \eqref{Ansatz} there exists a unique solution to the system at least up to a radius $R_0 + \Delta R$, where $R_0$ is the boundary of the support of the uncharged background solution and $\Delta R$ is a constant strictly bigger than zero. In addition, the matter quantities of the solution have support bounded by the radius $R_{0q}$ with $R_{0q}\leq R_0+\Delta R$.
  \end{lemma}

  \begin{proof}  

 Since

    \begin{equation}
        r_{-Q_0} \leq r_{-} \leq r_a \leq r_b \leq r_{+} \leq r_{+Q_0}
    \end{equation}
    
    we have

    \begin{equation}
        \max\{r_{-}, r_{-q}\} < \min\{r_{+}, r_{+q}\}
    \end{equation}
    
    So we can use the smaller radius of $r_{+}$ and $r_{+q}$ to perform the gluing process, because this will allow us to do the gluing at the same radius for the uncharged solution. For this purpose, we define

    \begin{equation} \label{r_0}
        r_0 := \min\{r_{+}, r_{+q} \}
    \end{equation}
    
    After integrating the Einstein equation \eqref{lambda diff} for $\lambda_q(r)$ from $r_0$ to $r$ and using the boundary condition

    \begin{equation}
        e^{-2\lambda_q(r_0)} = 1 - \frac{2M_0}{r_0} + \frac{Q_0^2}{r_0^2}
    \end{equation}
    
    we get
    
    \begin{align}
    \begin{split}
        e^{-2\lambda_q(r)} &= 1 -  \frac{8\pi}{r}\int_{r_0}^r s^2 g_q(s) ds - \frac{1}{r} \int_{r_0}^r \frac{q^2(s)}{s^2}ds + \frac{r_0}{r}\left( e^{-2\lambda_q(r_0)} -1 \right)\\
        &= 1 - \frac{2}{r} \left( \frac{r_0}{2} \left( 1- e^{-2 \lambda_q(r_0)} \right) + 4 \pi \int_{r_0}^r s^2 g_q(s) ds +\frac{1}{2} \int_{r_0}^r \frac{q^2}{s^2}ds \right) 
    \end{split}
    \end{align}
    
    Plugging this in to the second nontrivial Einstein equation \eqref{mu diff} for $\mu'(r)$ we end up with the reduced Einstein equation

    \begin{equation} \label{Einstein eq background}
        \mu_q'(r) = \frac{4 \pi r h_q(r) + \frac{r_0}{2r^2} \left( 1 - e^{-2 \lambda_q(r_0)}\right) + \frac{4\pi}{r^2} \int_{r_0}^r s^2 g_q(s)ds - \frac{q^2(r)}{2 r^3} + \frac{1}{2 r^2} \int_{r_0}^r \frac{q^2(s)}{s^2} ds }{1 - \frac{2}{r} \left( \frac{r_0}{2} \left( 1- e^{-2 \lambda_q(r_0)} \right) + 4 \pi \int_{r_0}^r s^2 g_q(s) ds +\frac{1}{2} \int_{r_0}^r \frac{q^2}{s^2}ds \right) } 
    \end{equation}

    for $r \geq r_0$ with boundary condition

    \begin{equation}
        e^{2 \mu_{0q}}:= e^{2 \mu_q(r_0)} = 1 - \frac{2 M_0}{r_0} + \frac{Q_0^2}{r_0^2}
    \end{equation}

    For the Maxwell equation \eqref{Maxwell} we impose the boundary condition $q(r_{BQ_0})=Q_0$. Since the right-hand side of the Einstein equation \eqref{Einstein eq background} is regular for all $r\geq r_0$, a solution exists locally. Furthermore, for the Einstein equation with regular center a continuation criterion has been shown in \cite{thaller}, which guarantees the existence of a solution at least as long as the denominator of the equation is positive. Since the equation in \eqref{Einstein eq background} only differs from the corresponding Einstein equation with regular center by some additional terms depending on $r_0$, $M_0$ and $Q_0$, which are however bounded and well behaved on any finite interval, the continuation criterion still applies. Now we proceed to show that the maximum radius of existence $R_C$ is bigger equal $R_0 + \Delta R$, with some finite and positive $\Delta R$, if we choose the absolute value of the charges $q_0$ and $Q_0$ small enough. We denote the denominator in the uncharged Einstein equation by

    \begin{equation}
        d_{M_0}(r) := 1- \frac{2}{r} \left( M_0 + 4 \pi \int_{r_0}^r s^2 g(s) ds \right)
    \end{equation}

    and in the charged case by
    
    \begin{equation}
        d_{M_0 Q_0}(r) := 1- \frac{2}{r} \left( M_0 + 4 \pi \int_{r_0}^r s^2 g_q(s) ds \right) + \frac{Q_0^2}{r_0 r} - \frac{1}{r} \int_{r_0}^r \frac{q^2(s)}{s^2}ds
    \end{equation}

    and define

    \begin{equation} \label{Delta d_0}
        \Delta d_0 := \alpha \cdot d_{M_0 Q_0}(r_0) = \alpha \left( 1 - \frac{2 M_0}{r_0} + \frac{Q_0^2}{r_0^2} \right)
    \end{equation}

  for a positive constant $\alpha < 1$. Furthermore, we define the radii

    \begin{equation}
        r^* := \inf \{ r \in (r_0, R_C) | d_{M_0 Q_0}(r) = \Delta d_0 \}
    \end{equation}

    and
    
    \begin{equation}
        \Tilde{r} := \inf \{ r \in (r_0, R_C) | \vert \gamma_q(r) - \gamma(r) \rvert > \lvert \gamma(R_0 + \Delta R) \rvert \}
    \end{equation}
    
    and their minimum

    \begin{equation}
        \Tilde{r}^* := \min\{r^*, \Tilde{r} \}
    \end{equation}

     and note that there exists a solution for all $r \in (r_0, \Tilde{r}^*)$, because the continuation criterion applies. So, we want to show that $\Tilde{r}^* \geq R_0 + \Delta R$. To do so we establish an upper bound on the distance between the denominators as well as on the distance between the two $\gamma$-functions. Both of them will follow from an estimate on $\lvert g_q -g \rvert + \lvert h_q -h \rvert$ which we will show in the following. We split the estimate up in two parts

     \begin{align} \label{split g_Lambda - g}
        \lvert g_q - g \rvert \leq \lvert \underbrace{g(r, \mu_q, \Tilde{I}_q) - g(r, \mu, \Tilde{I}_q) }_{=: \Delta_\mu g}\rvert + \lvert\underbrace{ g(r, \mu, \Tilde{I}_q) - g(r. \mu,0) }_{=:\Delta_{\Tilde{I}} g}\rvert
    \end{align}
     
    Similar to \cite{fajman} we obtain

    \begin{align}
    \begin{split}
        \lvert \Delta_{\Tilde{I}} g \rvert &\leq \underset{J\in \left[ I, I_\Lambda \right]}{\sup} \left\lvert \frac{\partial g(r, \mu, J)}{\partial J} \right\rvert \lvert I_\Lambda(r) - I(r) \rvert
    \end{split}
    \end{align}
    and
    \begin{align}
    \begin{split} \label{Delta_mu g}
        \lvert \Delta_\mu g \rvert &\leq \underset{u\in \left[ \mu, \mu_\Lambda \right]}{\sup}\left\lvert \frac{\partial g(r, u, I_\Lambda)}{\partial u} \right\rvert \lvert \mu_\Lambda - \mu \rvert\\
        &=: C(r) \lvert \mu_\Lambda - \mu \rvert
    \end{split}
    \end{align}
     
    and estimate the distance between $\mu_q(r)$ and $\mu(r)$ for $r_0 \leq r \leq \Tilde{r}^*$. For the uncharged solution we also glue the trivial Ansatz for the particle distribution function to a nontrivial one at $r_0$ and impose the conditions

     \begin{align}
         e^{-2 \lambda(r_0)} = 1 - \frac{2M_0}{r_0}
     \end{align}

     and 

     \begin{equation}
         e^{2\mu_0} := e^{2 \mu(r_0)} = 1 - \frac{2 M_0}{r_0}
     \end{equation}
     
     to match the solution to the non-shifted Schwarzschild metric at $r_0$. Firstly, we calculate the difference between the initial values of $\mu$ and $\mu_q$

    \begin{align}
    \begin{split}
        \lvert \mu_0 - \mu_{0q} \rvert &= \frac{1}{2} \ln \left( \left( 1 - \frac{2 M_0}{r_0} + \frac{Q_0^2}{r_0^2} \right) \left( 1 - \frac{2 M_0}{r_0} \right)^{-1} \right)\\
        &= \frac{1}{2} \ln \left( 1 + \frac{Q_0^2}{r_0^2} \left( 1 - \frac{2 M_0}{r_0} \right)^{-1} \right) =: C_{0 Q_0}
    \end{split}
    \end{align}

    where $C_{0 Q_0}$ is a constant growing in $Q_0^2$ or equivalently in $\lvert Q_0 \rvert$ and vanishing as $Q_0 \to 0$.  For computational convenience we set $g_q$, $g$, $h_q$ and $h$ equal to $0$ for $r \in (0, r_0)$. With this we estimate

    \begin{align}
    \begin{split} \label{mu - mu_Lambda background}
        \lvert \mu(r) - \mu_q(r) \rvert &\leq C_{0 Q_0} + \int_{r_0}^r \lvert \mu'(s) - \mu_q'(s) \rvert ds\\
        &\leq C_{0 Q_0} + \int_{r_0}^r \frac{1}{d_{M_0Q_0}(s)} \bigg( 4\pi s \lvert h_q(s) - h(s) \rvert\\
        &\hspace{0.4cm} + \frac{4 \pi}{s^2} \int_0^s \sigma^2 \left\lvert g_q(\sigma) - g(\sigma) \right\rvert d\sigma + \left\lvert \frac{Q_0^2}{2r_0s^2} \right\rvert\\
        & \left. \hspace{0.4cm} + \left\lvert - \frac{q^2}{2s^3} + \frac{1}{2s^2} \int_{r_0}^s \frac{q^2(\sigma)}{\sigma^2} d\sigma \right\rvert \right) ds\\
        &\hspace{0.4cm} + \int_{r_0}^r \left\lvert \frac{4 \pi}{d_{M_0 Q_0}(s)} - \frac{4 \pi}{d_{M_0}(s)} \right\rvert \left( s h(s) + \frac{1}{s^2} \int_0^s \sigma^2 g(\sigma) d\sigma + \frac{M_0}{s^2} \right) ds
    \end{split}
    \end{align}

    where we plugged in \eqref{Einstein eq background} for the charged an uncharged case respectively and then subtracted and added a term 
    
    $$\frac{1}{d_{M_0Q_0}(s)} \left(s h(s) + \frac{1}{s^2} \int_0^s \sigma^2 g(\sigma) d \sigma + \frac{M_0}{s^2} \right)$$
 
 inside of the absolute value under the integral. The terms with the differences between the $g$- and $h$-functions can be estimated to
    
    \begin{gather}
        \int_{r_0}^r \underbrace{\frac{1}{d_{M_0Q_0}(s)}}_{\leq \frac{1}{\Delta d_0}} \left( 4\pi s \lvert h_q(s) - h(s) \rvert + \frac{4 \pi}{s^2} \int_0^s \sigma^2 \left\lvert g_q(\sigma) - g(\sigma) \right\rvert d\sigma\right) ds\\
        \leq C(r) \int_{r_0}^r \left( \lvert h_q(s) - h(s) \rvert + \left\lvert g_q(s) - g(s) \right\rvert \right) ds
    \end{gather}
    
    with a positive function $C(r)$, growing in $r$ and independent of all charges. For the black hole charge term we get

    \begin{equation}
        \left\lvert \frac{Q_0^2}{2r_0 s^2} \right\rvert \leq \frac{Q_0^2}{2r_0^3} =:C_{Q_0}
    \end{equation}
    
    and for the other charge terms

    \begin{align}
    \begin{split} \label{q terms estimate}
        \left\lvert \int_0^r \frac{1}{2d_{M_0Q_0}(s)} \left( - \frac{q^2}{s^3} + \frac{1}{s^2} \int_0^s \frac{q^2(\sigma)}{\sigma^2} d \sigma \right) ds \right\rvert &\leq \frac{1}{2\Delta d_0} \left\lvert \int_0^r \frac{d}{ds} \left( \frac{1}{s} \int_0^s \frac{q^2(\sigma)}{\sigma^2}d\sigma \right) ds \right\rvert \\
        &= \frac{1}{2\Delta d_0} \frac{1}{r} \int_0^r \frac{q^2(s)}{s^2} ds\\
        &\leq (Q_0+q_0)^2 C(r)
    \end{split}    
    \end{align}
    
    with a positive, continuous function $C(r)$, growing in $r$ and again independent of $Q_0$ and $q_0$, where the last line follows because $d_{M_0Q_0}(r)>0$, so $\frac{1}{r} \int_0^r \frac{q^2(s)}{s^2} ds$ remains bounded. The terms in the brackets of the last line in estimate \eqref{mu - mu_Lambda background} can also be estimated by some $C(r)$ with the same properties as the other $C(r)$ we had before. This follows from $h(s), g(s) \leq C(r)$ for $s \leq r$ which has been shown in \cite{thaller} and from $\frac{M_0}{s^2} \leq \frac{M_0}{r_0^2}$. It remains to estimate the distance between the inverse values of the denominators. Since we have Vlasov matter with $p + 2p_T = \rho$, the generalized Buchdahl inequality with Schwarzschild singularity at the center holds for all $r \in [\frac{9M_0}{4}, \infty)$ (see \cite{fajman}, Lemma 5.1) which reads

    \begin{equation} \label{Buchdahl BH}
        \frac{2M_0 + 8 \pi \int_{r_0}^r s^2 \rho ds }{r} \leq \frac{8}{9}.
    \end{equation}
 The proof indeed generalizes to the EVMS.     
    Since

    \begin{align}
    \begin{split}
        r_0 &\geq \hat{r} = \frac{L_0}{2M_0} - \sqrt{\frac{L_0^2}{4M_0^2}-3L_0} \geq \frac{9M_0}{4}
    \end{split}
    \end{align}
    
   the generalized Buchdahl inequality \eqref{Buchdahl BH} holds in the EVS with central black hole on all $r \in [r_0, \infty)$. Using this result we get for the denominator of the uncharged system $d_{M_0}(r) \geq \frac{1}{9}$. Thus, we can estimate

    \begin{equation}
        \left\lvert \frac{1}{d_{M_0Q_0}} - \frac{1}{d_{M_0}} \right\rvert = \frac{\left\lvert d_{M_0} -d_{M_0Q_0} \right\rvert}{d_{M_0Q_0} d_{M_0}} \leq \frac{9 \left\lvert d_{M_0} - d_{M_0Q_0} \right\rvert}{\Delta d_0 }
    \end{equation}

   and calculate

    \begin{align}
    \begin{split}
        \lvert d_{M_0} (r) - d_{M_0 Q_0}(r) \rvert & \leq \frac{8 \pi}{r} \int_0^r s^2 \lvert g_q(s) - g(s) \rvert ds + \frac{1}{r} \int_{r_0}^r \left\lvert \frac{q^2(s)}{s^2} \right\rvert ds + \frac{Q_0^2}{r_0 r} \\
        &\leq 8 \pi r \int_0^r \lvert g_q(s) - g(s) \rvert ds + (q_0^2 + q_0 Q_0) C(r) + C_{Q_0}.
    \end{split} \label{d_M - d_MQ difference}
    \end{align}
    
    with $C(r)$ with the same properties as the other $C(r)$ we had before and a positive constant $C_{Q_0}$ that is independent of $r$ and $q_0$, growing in $\lvert Q_0 \rvert$ and converging to zero as $Q_0$ does. Collecting all the estimates, we obtain

    \begin{align}
    \begin{split} \label{mu - mu_q background}
        \lvert \mu(r) - \mu_q(r) \rvert &\leq r C_{Q_0} + (q_0^2 + q_0Q_0) C(r) + C(r) \int_0^r \left( \lvert g_q -g \rvert + \lvert h_q -h \rvert \right) ds.
    \end{split}
    \end{align}

    To get an upper bound for $\lvert g_q - g \rvert$ we are left with estimating $\lvert \Delta_{\Tilde{I}} g \rvert$. We find
    
    \begin{align}
    \begin{split} \label{I tilde prop}
        \Tilde{I}_q(r) &= q_0 \int_{r_v}^r \frac{q(s)}{s^2} e^{\mu(s) + \lambda(s)} ds\\
        &\propto q_0 Q_0 C(r) + q_0^2 C(r)
    \end{split}
    \end{align}
    
    which is positive because of the same sign of $Q_0$ and $q_0$ with two positive functions $C(r)$ that are independent of $q_0$ and $Q_0$. Thus, we can estimate

    \begin{align}
    \begin{split}
        \lvert \Delta_{\Tilde{I}} g \rvert &\leq (q_0Q_0 +q_0^2) C(r)
     \end{split}
    \end{align}
    
    and it results for the difference between the $g$- and $h$- functions that

    \begin{equation} \label{g h dist background 1}
        \lvert g_q -g \rvert + \lvert h_q -h \rvert \leq \left(q_0 Q_0 + q_0^2 \right) C(r) + rC_{Q_0}(r) + C(r) \int_0^r \left( \lvert g_q - g \rvert + \lvert h_q- h \rvert \right) ds
    \end{equation}

    where again the functions $C(r)$ can be different from the ones before, but they keep the same properties. A Gr\"onwall argument then implies

    \begin{equation} \label{g-g_q + h-h_q}
        \lvert g_q -g \rvert + \lvert h_q -h \rvert \leq (q_0 Q_0 + q_0^2) C(r) + r C_{Q_0}.
    \end{equation}

    Plugging this estimate in \eqref{mu - mu_q background} gives

    \begin{equation}
         \lvert \mu(r) - \mu_q(r) \rvert \leq (q_0 Q_0 + q_0^2) C(r) + C_{Q_0}(r)
    \end{equation}

    This also means for the distance between the $\gamma$-functions
    
    \begin{equation} \label{gamma dist back 1}
        \lvert \gamma_q(r) - \gamma(r) \rvert \leq (q_0 Q_0 + q_0^2) C(r) + C_{Q_0}(r)
    \end{equation}
    
    Thus, choosing the absolute value of the charges $q_0$ and $Q_0$ small enough we can assure

    \begin{equation} \label{gamma_q - gamma final sing}
         \lvert \gamma_q(r) - \gamma(r) \rvert \leq \lvert \gamma (R_0 + \Delta R) \rvert
    \end{equation}

    for all $r \in (r_0, R_0+ \Delta R)$. Looking at the distance \eqref{d_M - d_MQ difference} of the denominators and keeping in mind estimate \eqref{g-g_q + h-h_q} for the difference between the $g$-functions, we obtain

    \begin{equation} \label{d dist back 1}
        \lvert d_{M_0} (r) - d_{M_0 Q_0}(r) \rvert \leq (q_0 Q_0 + q_0^2) C(r) + C_{Q_0}(r)
    \end{equation}
    
    and see that by choosing $\lvert q_0 \rvert$ and $\lvert Q_0 \rvert$ small enough, we can reach on all $r \in (r_0, R_0+ \Delta R)$

    \begin{equation}
        \lvert d_{M_0} (r) - d_{M_0 Q_0}(r) \rvert \leq \frac{1}{9} - \Delta d_0
    \end{equation}

    where we choose the constant $\alpha$ in \eqref{Delta d_0} small enough, such that $\frac{1}{9} - \Delta d_0 \geq C >0 $. Due to the Buchdahl inequality we have $d_{M_0} \geq \frac{1}{9}$ and consequently we get

    \begin{equation}
        \Delta d_0 \leq d_{M_0 Q_0}(r)
    \end{equation}

    for all $r_0 \leq r \leq R_0 + \Delta R$ which, together with \eqref{gamma_q - gamma final sing} means that

    \begin{equation}
        \Tilde{r}^* \geq R_0 + \Delta R.
    \end{equation}

    This implies that a solution $(\mu_q, \lambda_q, q)$ exists at least until $R_0 + \Delta R$ and that the matter quantities are of bounded support with boundary $R_{0 q} \leq R_0 + \Delta R$. 
  \end{proof}
  
    \begin{proof} [Proof of the main theorem]
    Beyond the support there is again a vacuum region, where we can glue the nontrivial Ansatz \eqref{Ansatz} (with the modified energy) for the particle contribution function continuously to the constant zero function and thus continue the solution by a correctly shifted Reissner-Nordstr\"om solution with total mass parameter $M$ which corresponds to the sum of the black hole mass $M_0$ and the mass $4 \pi \int_{r_0}^{R_{0q}} s^2 g_q(s)ds$ of the matter shell and with total charge parameter $Q$ which corresponds to the sum of the black hole charge $Q_0$ and the total charge $4 \pi \int_{r_0}^{R_{0q}} s^2 \rho_q(s) ds$ of the matter region.
\end{proof}

\section{Discussion}
The static solutions constructed in the present work generalize similar static solutions to the EVMS in previous works to the case of a black hole in the center. We expect that it is possible to generalize this class as well as the regular solutions constructed in \cite{thaller} to the case of a positive cosmological constant. This will be addressed in a future paper. Similarly, we intend to consider multiple-shell solutions consisting of particles of potentially different charges.

\section{Appendix}

We collect calculations that have not been documented in detail in the main body of the paper.

\subsection{Proof that $\Tilde{D} > 0$} \label{A1}

To show that the equation $a_{Q_0}(r)=1$ has three distinct, real solutions we need to show that $\Tilde{D}$, defined in \eqref{discriminant} is strictly positive. For this aim, at first we expand the brackets and find

\begin{align}
\begin{split}
    \Tilde{D} &= 4M_0^2L_0 \left(-2Q_0^4 + 5Q_0^2L_0-4M_0^2L_0 \right) +M_0^2L_0^3 - Q_0^2 \left(Q_0^2+L_0 \right)^3\\
    &= -8 L_0 M_0^2 Q_0^4 -16 L_0^2 M_0^4 - Q_0^8 - L_0^3 Q_0^2 -3 L_0 Q_0^6 - 3 L_0^2 Q_0^4 + 20 L_0^2 M_0^2 Q_0^2 + L_0^3 M_0^2
\end{split}
\end{align}

then we split up the second last term to remove step by step the negative terms. We write

\begin{equation}
    20 L_0^2 M_0^2 Q_0^2 = 3 L_0^2 M_0^2 Q_0^2 + 1 L_0^2 M_0^2 Q_0^2 + 16 L_0^2 M_0^2 Q_0^2
\end{equation}

and do the following estimates

\begin{align}
\begin{split}
    L_0^2 M_0^2 Q_0^2 &> 16 L_0 M_0^4 Q_0^2 > 16 L_0 M_0^2 Q_0^4 > 8 L_0 M_0^2 Q_0^4 + 8 L_0 Q_0^6\\
    &> 8 L_0 M_0^2 Q_04 + 3 L_0 Q_0^6  + 5 \cdot 16 M_0^2 Q_0^6 > 8 L_0 M_0^2 Q_04 + 3 L_0 Q_0^6  + 5 \cdot 16 Q_0^8\\ &> 8 L_0 M_0^2 Q_04 + 3 L_0 Q_0^6  + Q_0^8
\end{split}
\end{align}

and

\begin{equation}
     3 L_0^2 M_0^2 Q_0^2 > 3 L_0^2 Q_0^4
\end{equation}

This means for $\Tilde{D}$ that

\begin{align}
\begin{split}
    \Tilde{D} &> - 16 L_0^2 M_0^4 - L_0^3 Q_0^2 + 16 L_0^2 M_0^2 Q_0^2 + L_0^3 M_0^2\\
    &= -16L_0^2M_0^2 \left( M_0^2 - Q_0^2 \right) + L_0^3 \left( M_0^2 - Q_0^2 \right)\\
    &= L_0^2 \underbrace{\left( M_0^2 - Q_0^2 \right)}_{>0} \underbrace{\left( L_0 - 16 M_0^2 \right)}_{>0} > 0
\end{split}
\end{align}

and the proof is complete.

\subsection{Proof that $r_{-Q_0} > M_0 - \sqrt{M_0^2-Q_0^2}$} \label{A2}

To prove that $r_{-Q_0} > M_0 - \sqrt{M_0^2-Q_0^2}$, it is sufficient to proof the inequality in \eqref{rhs > r0}, because we already realized in \eqref{for apx A2} that $r_{-Q_0} \geq \frac{Q_0^2 + L_0 - \sqrt{\left(Q_0^2+L_0\right)^2 - 12L_0M_0^2}}{6M_0}$. To proof that this latter expression is bigger than $M_0-\sqrt{M_0^2-Q_0^2}$, we firstly realize that

\begin{align} \label{final ineq r-Q0 r0}
    12M_0^2 \left(3M_0^2-Q_0^2 \right) > 0
\end{align}

and that

\begin{equation}
    12M_0^2 \left(3M_0^2-Q_0^2 \right) = 12L_0 M_0^2 - \left(L_0 + Q_0^2 \right)^2 + \left(L_0 + Q_0^2 \right)^2 -12 M_0^2 \left(L_0 + Q_0^2 \right) + 36M_0^4
\end{equation}

plugging this equality into inequality \eqref{final ineq r-Q0 r0} and bringing the first two terms to the left-hand side, we obtain

\begin{align} \label{ineq 1}
    \left(L_0 + Q_0^2 \right)^2 -12 M_0^2 \left(L_0 + Q_0^2 \right) + 36M_0^4 > \left(L_0 + Q_0^2 \right)^2 - 12L_0 M_0^2
\end{align}

Since

\begin{equation}
    \left(L_0 + Q_0^2 \right)^2 - 12 L_0 M_0^2 > L_0^2 - 12L_0 M_0^2 > 16 L_0 M_0^2 - 12L_0 M_0^2 > 0
\end{equation}

we can take the square root of inequality \eqref{ineq 1}. For this purpose we note that

\begin{equation}
    \left(L_0 + Q_0^2 \right)^2 -12 M_0^2 \left(L_0 + Q_0^2 \right) + 36M_0^4 = \left( L_0 + Q_0^2 - 6 M_0^2 \right)^2
\end{equation}

and find

\begin{align}
\begin{split} \label{ineq 2}
     L_0 + Q_0^2 - 6M_0^2 &> \sqrt{\left(L_0 + Q_0^2 \right)^2 - 12 L_0 M_0^2}\\
     &> \sqrt{\left(L_0 + Q_0^2 \right)^2 - 12 L_0 M_0^2} - 6M_0 \sqrt{M_0^2 - Q_0^2}
\end{split}
\end{align}

We rearrange this to

\begin{equation} \label{ineq 3}
    L_0+Q_0^2 - \sqrt{\left(L_0 + Q_0^2 \right)^2 - 12 L_0 M_0^2} > 6 M_0^2 - 6M_0 \sqrt{M_0^2 - Q_0^2} 
\end{equation}

and divide as a last step by $6M_0$ which results in

\begin{equation}
    \frac{L_0+Q_0^2 - \sqrt{\left(L_0 + Q_0^2 \right)^2 - 12 L_0 M_0^2}}{6M_0} > M_0 - \sqrt{M_0^2 - Q_0^2}
\end{equation}

and proves \eqref{rhs > r0} and thus $r_{-Q_0} > M_0 - \sqrt{M_0^2-Q_0^2}$.

\subsection{Proof that $r_{+Q_0} \geq r_{-}$} \label{A3}

In order to show that $r_{+Q_0} \geq r_{-}$, we estimate

\begin{align}
\begin{split} \label{r+Q0 - r-}
    r_{+Q_0} - r_{-} & \geq \frac{Q_0^2 + L_0}{6M_0} + \frac{1}{2} \sqrt{\frac{4}{3} \left( \frac{\left( Q_0^2 + L_0 \right)^2}{12 M_0^2}-L_0 \right)} - \frac{L_0}{4M_0} + \sqrt{\frac{L_0^2}{16M_0^2} - L_0}\\
    & \geq \frac{L_0}{6M_0} + \frac{1}{2} \sqrt{\frac{L_0^2}{9 M_0^2}- \frac{4L_0}{3}} - \frac{L_0}{4M_0} + \sqrt{\frac{L_0^2}{16M_0^2} - L_0}\\
    & \geq \frac{1}{2} \sqrt{\frac{L_0^2}{9 M_0^2} - \frac{4L_0}{3}} - \frac{L_0}{12 M_0}
\end{split}
\end{align}

Since

\begin{align}
\begin{split}
    \frac{1}{4} \left( \frac{L_0^2}{9M_0^2} - \frac{4L_0}{3} \right) - \frac{L_0^2}{144M_0^2} &= \frac{L_0^2}{36M_0^2} - \frac{L_0^2}{144M_0^2} - \frac{L_0}{3} = \frac{3L_0^2}{144M_0^2} - \frac{L_0}{3}\\
    &> \frac{48M_0^2L_0}{144M_0^2} - \frac{L_0}{3} = \frac{L_0}{3 } - \frac{L_0}{3} = 0
\end{split}
\end{align}

we have

\begin{equation}
    \frac{1}{4} \left( \frac{L_0^2}{9M_0^2} - \frac{4L_0}{3} \right) > \frac{L_0^2}{144M_0^2}
\end{equation}

Due to positivity we can take the square root which leads to

\begin{equation}
    \frac{1}{2} \sqrt{\frac{L_0^2}{9M_0^2} - \frac{4L_0}{3}} > \frac{L_0}{12M_0}
\end{equation}

Using this in inequality \eqref{r+Q0 - r-} yields

\begin{equation}
    r_{+Q_0} - r_{-} > 0
\end{equation}

and we are done.

\bibliography{references}
\bibliographystyle{plainnat}
\end{document}